\documentclass[11pt]{paper}

\usepackage{amsmath}
\usepackage{amsthm}
\usepackage{amssymb}
\usepackage{url}
\usepackage{todonotes}

\usepackage[left=1in,right=1in,top=1in,bottom=1in]{geometry}

\newcommand{\NP}{\textsf{NP}}

\newcommand{\Opt}{\textsf{Opt}}

\renewcommand{\b}[1]{\mathbf{#1}}

\newcommand{\MaxCSP}{\textsc{Max CSP}}
\newcommand{\MaxPartCSP}{\textsc{Max PartCSP}}
\newcommand{\R}{\mathbb{R}}
\newcommand{\scalprod}[1]{\left<  #1 \right> }

\DeclareMathOperator{\Odd}{Odd}
\DeclareMathOperator{\sign}{sign}
\DeclareMathOperator{\poly}{poly}
\DeclareMathOperator{\supp}{supp}
\DeclareMathOperator*{\E}{\mathbb{E}}
\DeclareMathOperator*{\Var}{Var}

\DeclareMathOperator{\Inf}{Inf}

\newtheorem{theorem}{Theorem}[section]
\newtheorem{lemma}[theorem]{Lemma}
\newtheorem{proposition}[theorem]{Proposition}

\newtheorem{conjecture}[theorem]{Conjecture}

\newtheorem{claim}[theorem]{Claim}
\newtheorem{fact}[theorem]{Fact}

\theoremstyle{remark}

\theoremstyle{definition}
\newtheorem{definition}[theorem]{Definition}

\title{A Characterization of Approximation Resistance for Even
  $k$-Partite CSPs}
\author{
  Per Austrin\thanks{Work done while at the University of Toronto supported by NSERC.}\\Aalto University and KTH Royal Institute of Technology\\
  \and
  Subhash Khot\thanks{Research supported by NSF Expeditions grant CCF-0832795, NSF Waterman Award and BSF grant 2008059.}\\New York University and University of Chicago
}

\begin{document}

\maketitle

\begin{abstract}
  A constraint satisfaction problem (CSP) is said to be
  \emph{approximation resistant} if it is hard to approximate better
  than the trivial algorithm which picks a uniformly random
  assignment.  Assuming the Unique Games Conjecture, we give a
  characterization of approximation resistance for $k$-partite CSPs
  defined by an even predicate.
\end{abstract}

\thispagestyle{empty}
\addtocounter{page}{-1}
\clearpage

\section{Introduction}

In the past 20 years, there has been a significant amount of work done
on understanding the approximability of various \emph{constraint
  satisfaction problems} (CSPs).

For the purposes of this paper, a CSP is defined by a $k$-ary
predicate $P: \{-1,1\}^k \rightarrow \{0,1\}$ over a Boolean alphabet.\footnote{As is common, 
the input bits are written in the $\{-1,1\}$ notation with $-1$ interpreted as logical 
{\sf True} and $1$ as logical {\sf False}. Also ``parity" corresponds to taking product of the 
bits: odd parity means the product is $-1$ and even parity means the product is $1$.}  
An instance consists of a set of \emph{constraints}, each of which
dictates that $P$ applied to some list of $k$ literals should be
satisfied (a literal is a variable or the negation of a variable).
The objective is to find an assignment to the variables so as to
maximize the number of satisfied constraints.
Two well-known examples are \textsc{Max $k$-Sat} (where $P$
is the disjunction of the $k$ input bits) and \textsc{Max $k$-Lin}
(where $P$ is the parity of the $k$ input bits).

Essentially every \textsc{Max CSP} is \NP{}-hard (the exception being when $P$
only depends on one of its input bits).  In terms of approximability,
it is easy to see that choosing a uniformly random assignment to the variables,
without even looking at the instance, yields an approximation ratio of
$|P^{-1}(1)|/2^k$, where $|P^{-1}(1)|$ is the number of inputs in
$\{-1,1\}^k$ that satisfy $P$.

Improving upon this trivial algorithm turns out to be surprisingly
difficult.  In a groundbreaking paper, Goemans and Williamson
\cite{GW95} used semidefinite programming (SDP) to give improved
approximation algorithms for \textsc{Max 2-Sat} and \textsc{Max
  2-Lin}.  SDP was soon used to give better approximation algorithms
for many other problems as well, but for some CSPs, perhaps most
prominently \textsc{Max 3-Sat} and \textsc{Max 3-Lin}, no improvement
over the random assignment algorithm was found.  Then, in a new
breakthrough, H{\aa}stad \cite{Has01} showed that such an improvement
would not be possible: approximating \textsc{Max 3-Sat} within $7/8 +
\epsilon$ or \textsc{Max 3-Lin} within $1/2+\epsilon$ for some
$\epsilon > 0$ is \NP{}-hard.  In other words, \textsc{Max 3-Sat} and
\textsc{Max 3-Lin} have the remarkable property that the completely
mindless random assignment algorithm is optimal!

CSPs which have this property -- that they are \NP{}-hard to
approximate within $|P^{-1}(1)|/2^k + \epsilon$ -- are called
\emph{approximation resistant}.  Following H{\aa}stad's initial
result, many more CSPs have been shown to be approximation resistant
\cite{GLST98,ST00,EH08,Has05}.  Fairly quickly, a complete
characterization of approximation resistance for predicates of arity
three was found: $P: \{-1,1\}^3 \rightarrow \{0,1\}$ is approximation
resistant if and only if $P$ accepts all inputs of odd parity, or if
it accepts all inputs of even parity \cite{Has01,Zwi98}.

However, the next small case, predicates of arity $4$, is still not
completely classified, and it is not at all clear whether there is a
nice, clean characterization. We would like to emphasize that by a characterization we 
mean a necessary and sufficient condition.  
Modulo symmetries, there are $400$
non-constant predicates of arity $4$.  Hast \cite{Has05} showed $275$
of these to be approximable, $79$ of them to be approximation
resistant, and left the status of the remaining $46$ open.

In recent years, progress has been made on our understanding of
approximation resistance under the assumption of the \emph{Unique
  Games Conjecture} (UGC) \cite{Kho02}.  The first author and Mossel
\cite{AM09} proved that assuming the UGC, $P$ is approximation
resistant if there exists an unbiased pairwise independent
distribution over $\{-1,1\}^k$ supported on $P^{-1}(1)$.  Using this
condition, it can be shown that as $k \rightarrow \infty$, an
overwhelming fraction of all predicates are in fact approximation
resistant \cite{AH11}.  A somewhat more (complicated and) general sufficient 
condition is known \cite{AH12}.  As in \cite{AM09}, this condition is
in terms of the biases and pairwise correlations of distributions
supported on $P$.  At this point, it seems unlikely that there is a
clean characterization (necessary and sufficient), but one can hope that approximation resistance
is at least decidable.

Relevant here is the work of Raghavendra \cite{Rag08}, which
shows assuming the UGC that for any CSP, its approximability threshold
is determined by the integrality gap of a natural SDP relaxation for
the problem.  Furthermore, Raghavendra and Steurer \cite{RS09} show
that this integrality gap can be approximated to within an additive error
$\epsilon$ (in time doubly exponential in $\epsilon$).

This ``almost'' shows that it is decidable to determine whether a
CSP is approximation resistant.  However, as we have no a priori bound
on the error $\epsilon$ needed, it only shows that it is recursively
enumerable to determine whether a CSP is approximable.  Note that, for
every $k$ there is a smallest gap $\epsilon_k$ such that any
approximable predicate $P$ on $k$ bits can be approximated within at
least $|P^{-1}(1)|/2^k + \epsilon_k$.  If this number $\epsilon_k$ can
be computed, approximation resistance would be decidable, but it is
possible (though seemingly unlikely) that $\epsilon_k$ tends to $0$
faster than any computable function.

\subsection{Our Contribution}

The strength of \cite{Rag08}, namely that it works in a black-box
fashion for any CSP, is in some sense a weakness in this setting, as
it is not explicit and does not give any insight into what structural
properties cause a predicate to be approximation resistant.  In this
paper, we make progress towards an explicit characterization of
approximation resistance.  We restrict the class of CSPs we study in
two ways.
\begin{enumerate}
\item We only consider \emph{$k$-partite} instances.  In a $k$-partite
  instance, the variables are grouped into $k$ layers, and in each
  constraint, the literal passed as the $i$'th argument to $P$ comes
  from the $i$'th layer.
\item We only consider $P$ which are \emph{even}.  $P$ is even if
  $P(x) = P(- x)$ for every $x \in \{-1,1\}^k$, where $-x$ denotes
  bitwise negation of $x$.
\end{enumerate}

We refer to this as the $\MaxPartCSP(P)$ problem.
Our main contribution is an explicit necessary and sufficient characterization (assuming the UGC) of
when $\MaxPartCSP(P)$ is approximation resistant.  As in the case of
\cite{AM09} and its generalizations, our condition is based on the
existence of certain distributions $\mu$ over the set of satisfying
assignments of $P$ and furthermore the conditions on these
distributions depend only on their pairwise correlations
$\E_{\mu}[x_ix_j]$.

In order to properly state the characterization, we need to make a few
definitions.

\begin{definition}
  Let $G = (S,E)$ be a multigraph with vertex set $S \subseteq [k]$ and no self-loops. 
  For a correlation matrix $\rho \in \R^{k \times k}$ we define
  $\rho(G) = \prod_{ij \in E} \rho_{ij}$.  For a distribution
  $\Lambda$ over $k\times k$ correlation matrices we define
  $\Lambda(G) = \E_{\rho \in \Lambda}[\rho(G)]$.
\end{definition}

The key part of our definition is the existence of distributions
$\Lambda$ over correlation matrices each of which arises from a distribution over
$P^{-1}(1)$ -- we refer to these as $P$-supported correlation matrices -- such
that $\Lambda(G)$ vanishes on certain graphs.  Specifically:

\begin{definition}
  Let $\Lambda$ be a distribution over $k \times k$ correlation
  matrices.  We say that $\Lambda$ is \emph{$m$-vanishing on $P$} if:
  \begin{enumerate}
  \item $\Lambda$ is a distribution over $P$-supported correlation matrices.
  \item For every $S \subseteq [k]$ such that $\hat{P}(S) \ne 0$, and every
    odd-degree multigraph $G$ on $S$ with at most $m$ edges, it holds
    that $\Lambda(G) = 0$.
  \end{enumerate}
\end{definition}

Here $\hat{P}(S)$ denotes the Fourier coefficient of the predicate $P$ on the set $S$ 
(i.e.\ the coefficient of the monomial $\prod_{i\in S} x_i$ when $P$ is written as a multi-linear 
polynomial). Now we can state our main result.

\begin{theorem}
  \label{thm:main}
  Assuming the UGC, $\MaxPartCSP(P)$ is approximation resistant if and
  only if for every positive integer $m$ there exists a distribution $\Lambda$ which is
  $m$-vanishing on $P$.
\end{theorem}

Note that if there is a pairwise independent distribution supported on
$P$, i.e., if the identity matrix is $P$-supported, then taking
$\Lambda$ to be the singleton distribution on the identity matrix is
$m$-vanishing on $P$ for every $m$.  As such, this characterization 
generalizes the sufficient condition of \cite{AM09}.

Given $m$ and $P$, it is fairly easy to prove that the existence of an
$m$-vanishing $\Lambda$ on $P$ is decidable.  Hence the condition of
Theorem~\ref{thm:main} is recursively enumerable.  We feel that this
characterization is promising with respect to decidability.  For
instance, it is quite
possible that one can show some explicit upper bound on the largest
value of $m$ that one needs to check, which would immediately give
decidability. We also remark that, even though we do not prove it here, 
the characterization in Theorem \ref{thm:main} is equivalent to saying that there is a 
distribution $\Lambda$ which is $m$-vanishing on $P$ for all $m$ simultaneously.

\subsection{Proof Ideas}
\label{sec:proof-ideas}

We now briefly and informally outline the main ideas of the proof of
Theorem~\ref{thm:main}.

\paragraph{Algorithm.}

Suppose there is no $m$-vanishing distribution $\Lambda$ for some $m$.  By LP
duality, there are then constants $\{\gamma_G\}$ such that $\sum_{G}
\gamma_G \rho(G) > \delta$ for all $P$-supported $\rho$, where the sum
is over all odd-degree $G$ on at most $m$ edges.  Now, assume we are
given a solution to the basic SDP relaxation for $\MaxPartCSP(P)$ with
value $1$ (in reality it will only have value close to $1$ but this is
just a small technicality).  Then for each constraint we have a local
distribution $\mu$ and since the SDP value is $1$ its correlation
matrix $\rho$ is $P$-supported.  The basic idea is, very loosely, to
design a rounding algorithm which, given some graph $G$, finds an
assignment with value $|P^{-1}(1)|/2^k + \rho(G)$.  Picking a random
$G$ with probability proportional to $|\gamma_G|$ then gives an
assignment with value $|P^{-1}(1)|/2^k + \Omega(\delta)$.

To get an assignment with value $|P^{-1}(1)|/2^k + \rho(G)$, the idea
is to do as follows.  For simplicity, suppose $V(G) = [k]$ and
consider the monomial $\prod_{i=1}^k x_i$.  We can construct the
solution iteratively edge by edge, as follows.  Initially, set all
$x_i = 1$ (corresponding to the empty graph).  Then, for an edge $e =
(i,j)$, pick a standard Gaussian vector $\b{g}_e$, and multiply $x_i$
(resp.~$x_j$) by $\scalprod{\b{g}_e, \b{v}_i}$
(resp.~$\scalprod{\b{g}_e, \b{v}_j}$), where $\b{v}_i$ and $\b{v}_j$
are the vectors in the SDP solution corresponding to $x_i$ and $x_j$.
This operation has the effect of multiplying $\E[\prod_{i \in [k]}
  x_i]$ by a factor $\scalprod{\b{v}_i, \b{v}_j} = \rho_{ij}$ where
$\rho_{ij}$ is the correlation between $i$ and $j$ in the local
distribution on $x_1, \ldots, x_k$.  Repeating this for all edges of
the graph, we get $\E[\prod_{i \in [k]} x_i] = \rho(G)$, and we can
make sure that all other non-constant monomials have expectation $0$,
meaning that we get an advantage of $\rho(G)$ over $|P^{-1}(1)|/2^k +
\rho(G)$.  

To wrap this up and get the formal proof, there are som additional
technicalities to account for: the values assigned by the above
rounding are not Boolean-valued, we need to deal with negated
literals, and we need to take the magnitude of the Fourier
coefficients $\hat{P}(S)$ of $P$ into consideration.  The
formalization of the ``monomial rounding'' described above is given in
Lemma~\ref{lemma:round-layer} in Section~\ref{sec:round-layer} and its use to give a
non-trivial algorithm for $\MaxPartCSP(P)$ is then described in
Section~\ref{sec:full-algorithm}.

\paragraph{Hardness.} 

As is by now standard, the task of proving hardness boils down to
constructing a \emph{dictatorship test} using the predicate $P$.  The
dictatorship test gets oracle access to $k$ functions $f_1, \ldots,
f_k: \{-1,1\}^n \rightarrow \{-1,1\}$, and the question is whether
$f_1 = f_2 = \ldots = f_k$ are all equal to some dictatorship
function.  The test operates by picking $k$ inputs $x_1, \ldots, x_k$
and then accepts iff $P(f_1(x_1), f_2(x_2), \ldots, f_k(x_k)) = 1$.
The restriction to only using $P$ as the acceptance predicate is what
gives us hardness for $\MaxCSP(P)$ rather than an arbitrary CSP, and
the restriction that we have $k$ different functions and make one
query to each, instead of a single function, is precisely what gives
us hardness for $\MaxPartCSP(P)$ instead of $\MaxCSP(P)$.

Such a test is completely specified by the distribution of $(x_1,
\ldots, x_k)$.  To specify this we choose some very large $m$ and use
the $m$-vanishing distribution $\Lambda$ guaranteed to exist.  To
sample $(x_1, \ldots, x_k)$, we do as follows: first sample a
$P$-supported correlation matrix $\rho$ according to $\Lambda$, and
let $\mu$ be some $P$-supported distribution with correlation matrix
$\rho$.  Then, for each $i\in[n]$ we sample the $i$'th coordinate
$(x_1^i, \ldots, x_k^i)$ independently from $\mu$.  The completeness
of the test follows by $\mu$ being $P$-supported.  The soundness
follows using the invariance principle: first, we show that if the
functions $f_1, \ldots, f_k$ have low influence the acceptance
probability (appropriately arithmetized) can be well approximated by a
multilinear polynomial in Gaussian variables with the same second
moments as $x$.  Since higher moments of Gaussian variables are
determined by their covariance matrix, this multilinear polynomial
(and therefore also the acceptance probability) can be expressed as a
function of the covariance matrix, i.e., $\rho$, and it turns out that
all terms except for the constant $|P^{-1}(1)|/2^k$ are of the form
$\rho(G)$ for some odd-degree graph $G$ on less than $m$ edges.  Hence
taking the expectation over $\rho \sim \Lambda$, all non-constant
terms vanish.

\paragraph{Source of the various restrictions.}

It may be instructive to point out where the various restrictions we
impose come into play.

\begin{description}
\item[$k$-partiteness.]
  The fact that we know for each variable what ``role'' it will play
  is critical in allowing us to obtain the algorithm.  In particular,
  in the ``monomial rounding'' described above, it is important that
  any given variable corresponds to some given vertex of the graph $G$
  that we are using (the vertices of $G$ correspond to layers of the
  CSP instance).  If a vertex could appear as several different
  vertices of $G$ (i.e., in several different layers), it is not clear
  how to round it in such a way that the different occurences don't
  interfere with each other.

\item[Even predicates.]  This restriction is in some sense minor and
  more technical in nature.  It allows us to assume that the
  distributions $\mu$ supported on $P^{-1}(1)$ are unbiased, which
  simplies may arguments.  That said, it is not clear exactly how to
  generalize the present characterization to a general $P$.

\item[Odd-degree graphs.]  The reason why the characterization only
  involves odd-degree graphs is essentially the presence of negated
  literals.  First, in the algorithm it turns out that it is
  \emph{necessary} for the graphs to be odd-degree, as this
  essentially ensures that we don't have cancellations when dealing
  with negated literals.  Second, in the hardness result it turns out
  that it is \emph{sufficient} for the graphs to have odd-degree,
  because the functions $f_i$ we are testing can be assumed to be odd
  by the standard technique of folding, which is implemented by
  introducing negated literals.

\end{description}

\subsection{Discussion}

\paragraph{On the unnecessity of pairwise independence.}

It is known that there are approximation resistant predicates which do
not support a pairwise independent distribution.  A basic such example
is the predicate $GLST: \{-1,1\}^4 \rightarrow \{0,1\}$ defined by
$$
GLST(x_1, x_2, x_3, x_4) = \left\{\begin{array}{ll}
x_2 \ne x_3 & \text{if $x_1 = -1$}\\
x_2 \ne x_4 & \text{if $x_1 = 1$}
\end{array}\right..
$$ 
This predicate was shown to be approximation resistant by Guruswami et
al.\ \cite{GLST98}, but there is no pairwise independent distribution
supported on its accepting assignments -- indeed it is not difficult
to check that $x_2x_3+x_2x_4+x_3x_4 <0$ for all accepting inputs.  In
\cite{AH12}, Theorem VIII.6, a generalization of the pairwise
independence condition was given which also covers the $GLST$
predicate and in fact as far as we are aware cover all currently known
examples of approximation resistant predicates.

The condition of Theorem~\ref{thm:main} \emph{essentially} generalizes
the condition of \cite{AH12}.  We say ``essentially'' because Theorem
VIII.6 of \cite{AH12} in some cases allows for a condition referred to
as $\{i,j\}$-negativity, and it is not clear that this condition is
captured by Theorem~\ref{thm:main}.  The only example given in
\cite{AH12} using the $\{i,j\}$-negativity condition is not an even
predicate, so it is possible that this is a distinction between even
$P$ and general $P$.  On the other hand, it appears that for the
example given in \cite{AH12}, one can prove approximation resistance
without using $\{i,j\}$-negativity so it is not completely clear
whether allowing this adds any new predicates.  Another possibility is
that this is a distinction between $\MaxCSP(P)$ and $\MaxPartCSP(P)$,
because the proof in \cite{AH12} that $\{i,j\}$-negativity suffices
does not extend to partite instances.  In short, the situation is a
bit of a mystery and may warrant further study.

\paragraph{On $\MaxPartCSP(P)$ vis-\`{a}-vis $\MaxCSP(P)$.}

It is not known whether $\MaxPartCSP(P)$ behaves differently
from $\MaxCSP(P)$ with respect to approximation resistance.  Almost
all proofs of approximation resistance for $\MaxCSP(P)$, including
\NP-hardness results such as \cite{Has01, EH08}, can be adjusted to produce
$k$-partite instances, thereby showing approximation resistance for
$\MaxPartCSP(P)$.

However, one exception is the result of Raghavendra \cite{Rag08},
where it is not at all clear how to achieve this.  If it were the case
that the reduction of \cite{Rag08} can be adjusted to produce partite
instances, our restriction to $k$-partite instances would have been
without loss of generality (assuming the UGC), but as matters stand,
this can not be deduced.

Another exception is the hardness derived in \cite{AH12} from the
$\{i,j\}$-negativity condition mentioned above.

\subsection{Outline}

In Section~\ref{sec:notation} we introduce notation and terminology
used throughout the paper and state some known theorems that we need.
In Section~\ref{sec:deciding condition} we describe how to decide
whether an $m$-vanishing distribution exists.  We then proceed to
prove Theorem~\ref{thm:main}, giving an algorithm in
Section~\ref{sec:algorithm} and proving hardness in
Section~\ref{sec:hardness}. 

\section{Notation and definitions}
\label{sec:notation}

As is common, for convenience of notation we use $\{-1,1\}$ for
Boolean values rather than $\{0,1\}$.  Throughout, $P$ denotes a
$k$-ary predicate $P: \{-1,1\}^k \rightarrow \{0,1\}$ which we assume
to be \emph{even}, i.e., $P(x) = P(-x)$ for all $x$.

We say a distribution $\mu$ over $\{-1,1\}^k$ is \emph{$P$-supported}
if $\supp(\mu) \subseteq P^{-1}(1)$.  Similarly a correlation matrix
$\rho \in \R^{k \times k}$ is $P$-supported if there is a
$P$-supported $\mu$ such that $\rho_{ij} = \E_\mu[x_i x_j]$ for all
$i,j$.  Note that since $P$ is even, any $P$-supported distribution
can without loss of generality be assumed to be \emph{unbiased}, i.e.,
satisfying $\E_\mu[x_i] = 0$ for all $i$, as far as its correlation matrix is concerned
(since we can spread the probability mass equally on any pair of assignments $x$ and $-x$ without affecting
the correlation matrix).   

For the purposes of this paper, a \emph{multigraph} is a graph $G$
which may have multiple edges but no self-loops.  A multigraph has
\emph{odd degree} if every vertex of the graph has odd degree (when
edges are counted with multiplicities).  A key role in our
characterization is played by multigraphs $G$ whose vertex set is some
subset $S \subseteq [k]$, we refer to this as a \emph{multigraph on
  $S$}.

We write $S^n$ for the $n$-dimensional unit sphere (i.e., the set of
unit vectors in $\R^{n+1}$, and for two vectors $\b{u}, \b{v} \in
\R^n$ we write $\scalprod{\b{u}, \b{v}}$ for their standard inner
product.

\subsection{Partite Max-CSP and its SDP relaxation}

An instance $\Psi$ of $\MaxPartCSP(P)$ has $k \cdot n$ Boolean
variables indexed by $[k] \times [n]$.  Each constraint is of the form
$P(b_1 x_{1,j_1}, b_2 x_{2,j_2}, \ldots, b_k x_{k,j_k})$ for some
indices $j_1, \ldots, j_k$ and some signs $b_1, \ldots, b_k \in
\{-1,1\}$.

We use the following notation.  The constraints of an instance are
$(T_1, P_1)$, $(T_2, P_2)$, $\ldots$, where $T_i \subseteq [k] \times
[n]$ are the set of variables that the $i$'th constraint depends on --
exactly one from each layer -- and $P_i: \{-1,1\}^{T_i} \rightarrow
\{0,1\}$ is $P$ applied to the variables of $T_i$, possibly with some
variables negated.

We say that $\Psi$ is \emph{$\alpha$-satisfiable} if there is an assignment
to the variables which satisfies an $\alpha$ fraction of all the
constraints.

The basic SDP relaxation is described in Figure~\ref{fig:sdp}.     It has
as variables a vector $\b{v}_{i,j} \in S^{n \cdot k}$ for every
variable $x_{i,j}$, and an unbiased distribution $\mu_i$ over $\{-1,1\}^{T_i}$
for each constraint $(T_i, P_i)$. 
The fact that this is a relaxation follows from the
following observation: for any global integral assignment $\sigma \in \{-1,1\}^{k \cdot n}$, let ${\cal D}$ be the uniform distribution over
the pair of integral assignments $\sigma$ and $-\sigma$. Let $\mu_i$ be the restrictions of ${\cal D}$ to the respective sets $T_i$  
and ${\bf v}_{i,j} = \sigma_{i,j}$ be a $1$-dimensional vector. 
Then it is easy to see that this is a feasible solution to the SDP and its objective is same as the
fraction of constraints satisfied by $\sigma$ (or $-\sigma$). Here we use the
evenness of the predicate.

\begin{figure}[h]
  \framebox{
    \parbox{0.95\textwidth}{
      \begin{align*}
      \textbf{Maximize } & \sum_{i} \E_{x \sim \mu_{i}}[P_i(x)] \\
      \textbf{Subject to } & \textrm{$\mu_i$ is an unbiased distribution over $\{-1,1\}^{T_i}$} && \textrm{for every $i$} \\
      & \b{v}_{i,j} \in S^{n \cdot k} && \textrm{for all $(i,j) \in [k] \times [n]$}\\
      & \mu_i|_T = \mu_j|_T && \textrm{where $T = T_i \cap T_j$} \\
      & \scalprod{\b{v}_{i_1,j_1}, \b{v}_{i_2,j_2}} = \E_{x \sim \mu_l}[x_{i_1,j_1}x_{i_2,j_2}] && \textrm{for all $(i_1,j_1), (i_2,j_2) \in T_l$}
      \end{align*}
    }
  }
  \caption{SDP relaxation of $\MaxPartCSP(P)$.}
  \label{fig:sdp}
\end{figure}

\subsection{The Unique Games Conjecture}

In this section, we state the formulation of the Unique Games
Conjecture that we will use.

\begin{definition}\label{def:ug}
  An instance $\Lambda = (U, V, E, \Pi, [L])$ of \emph{Unique Games}
  consists of an unweighted bipartite multigraph $G = (U \cup V, E)$,
  a set $\Pi$ of \emph{constraints}, and a set $[L]$ of \emph{labels}.
  For each edge $e \in E$ there is a constraint $\pi_e \in \Pi$, which
  is a permutation on $[L]$.  The goal is to find a \emph{labeling}
  $\ell: U \cup V \rightarrow [L]$ of the vertices such that as many
  edges as possible are satisfied, where an edge $e = (u,v)$ is said
  to be satisfied by $\ell$ if $\ell(v) = \pi_e(\ell(u))$.
\end{definition}

\begin{definition} Given a Unique Game instance  $\Lambda = (U, V, E, \Pi, [L])$,
let $\Opt(\Lambda)$ denote the maximum
  fraction of simultaneously satisfied edges of $\Lambda$ by any
  labeling, i.\,e.,
$$ \Opt(\Lambda) := \frac{1}{|E|} \max_{\atop{ \ell: U \cup V
    \rightarrow [L]}} |\{\,e\,:\,\textrm{$\ell$ satisfies $e$}\,\}|. $$
\end{definition}

\begin{conjecture}(\cite{Kho02}) \label{conj:ugc}
  For every $\gamma > 0$, there is an integer  $L$ such that, for Unique Games
  instances $\Lambda$ with label set $[L]$ it is \NP-hard to
  distinguish between
  \begin{itemize}
  \item $\Opt(\Lambda) \ge 1-\gamma$
  \item $\Opt(\Lambda) \le \gamma$.
  \end{itemize}
\end{conjecture}

\subsection{Analytic Tools}

Any Boolean function $f: \{-1,1\}^n \rightarrow \R$ can be written
uniquely as a multilinear polynomial
\[
f(x) = \sum_{T \subseteq [n]} \hat{f}(T) \chi_T(x),
\]
where $\hat{f}(T)$ are the Fourier coefficients of $f$ and $\chi_T(x)
= \prod_{i \in T} x_i$.  As such, $f$ can be viewed as a multilinear
polynomial $f: \R^n \rightarrow \R$ and this is the view we commonly
take.  We write $f^{\le d}$ for the part of $f$ that is of degree $\le
d$, i.e., $f^{\le d}(x) = \sum_{|S| \le d} \hat{f}(S) \chi_S(x)$.
\begin{fact}
  \begin{align*}
    \E[f(X)] &= \hat{f}(\emptyset) &
    \Var[f(X)] &= \sum_{T \ne \emptyset} \hat{f}(T)^2,
  \end{align*}
  where the expectations are over a uniform $X$ in $\{-1,1\}^n$.
\end{fact}

\begin{definition}
  The \emph{influence} of the $i$'th variable on $f$ is
  \[
  \Inf_i(f) = \sum_{T \ni i} \hat{f}(T)^2,
  \]
  and the low-degree influence is
  \[
  \Inf^{\le d}_i(f) = \Inf_i(f^{\le d}) = \sum_{\substack{T \ni i\\|T| \le d}} \hat{f}(T)^2.
  \]
\end{definition}

As is common, the main analytic tool in our hardness result is the
invariance principle \cite{MOO10,Mos10}.  In particular, we have the
following theorem.

\begin{theorem}
  \label{thm:invariance}
  For every $k$, $\epsilon > 0$ there is a $\delta > 0$ such that the following
  holds.

  Let $\mu$ be an unbiased distribution over $\{-1,1\}^k$ with
  $\min_{x \in \{-1,1\}^k} \mu(x) \ge \epsilon$, $X$ be a random $k
  \times n$ matrix over $\{-1,1\}$ with each column distributed
  according to $\mu$, independently, and $G$ be a random $k \times n$
  matrix of standard Gaussians with the same covariance structure as
  $X$.

  Then for any $k$ multilinear polynomials $f_1, \ldots, f_k: \R^n
  \rightarrow \R$ with $\Inf_j^{\le 1/\delta}(f_i) \le \delta$ and $\Var[f_i] \le 1$ for all
  $i \in [k]$, $j \in [n]$, we have
  \[
  \left|\E_X\left[\prod_{i =1}^k f_i(X_i)\right] - \E_G\left[\prod_{i = 1}^k f^{\le 1/\delta}_i(G_i)\right]\right| \le \epsilon.
  \]
\end{theorem}

Theorem~\ref{thm:invariance} can be derived using Theorem 4.2 and
Lemma 6.2 of \cite{Mos10}: using Lemma 6.2 it follows that
$\E_X\left[\prod_{i=1}^k f_i(X_i)\right]$ is close to
$\E_X\left[\prod_{i=1}^k f^{\le 1/\delta}_i(X_i)\right]$ for
sufficiently small $\delta$.  Then we can use Theorem 4.2 on the
functions $\{f_i^{\le 1/\delta}\}$.  In Theorem 4.2, the values of
$f_i^{\le 1/\delta}(X_i)$ and $f_i^{\le 1/\delta}(G_i)$ are truncated
to the range $[0,1]$.  By scaling this holds with $[0,1]$ replaced by
some interval $[-B,B]$.  It is not too hard to show that for
sufficiently large $B$ (as a function of $k$ and $\epsilon$)
truncation to the interval $[-B, B]$ does not change $\E[\prod_{i=1}^k
  f^{\le 1/\delta}_i(X_i)]$ (resp.~$\E[\prod_{i=1}^k f^{\le
    1/\delta}_i(G_i)]$) by more than $\epsilon$.

\subsection{Products of Gaussians}

We need the following Lemma about the expectation of a product of
gaussians in terms of their pairwise correlations.

\begin{lemma}
  \label{lemma:gaussproduct}
  Let $r$ be an integer and let $g_1, \ldots, g_r$ be gaussians
  with mean $0$, variance $1$, and covariance matrix $\rho$.  Then
  $$
  \E\left[\prod_{i=1}^r g_i\right] = \sum_{M} \prod_{ij \in M} \rho_{ij},
  $$ where $M$ ranges over all perfect matchings of the complete graph
  on $r$ vertices (if $r$ is even there are $(r-1)!!$ terms and if $r$
  is odd the expectation is $0$).
\end{lemma}

\section{Decidability of $m$-vanishing Distributions}
\label{sec:deciding condition}

\begin{proposition}
  Given $P$ and $m$, the existence of a $\Lambda$ which
  is $m$-vanishing on $P$ is decidable.
\end{proposition}

\begin{proof}
  There are $M \leq 2^k \cdot k^{2m}$ graphs $G_1,\ldots,G_M$ of interest namely 
  graphs with at most $m$ edges supported on some vertex set $S \subseteq [k]$. We need to decide whether
  there is a distribution $\Lambda$ on $P$-supported correlation matrices such that 
  $$ \E_{\rho \in \Lambda} [ (\rho(G_1), \ldots, \rho(G_M)) ] = {\bf 0}. $$
  By Carath\'{e}odory's theorem this implies that $\Lambda$ can be
  assumed to have support at most $M$.  Each $\rho$ in the support of
  $\Lambda$ can be represented using $|P^{-1}(1)| + k^2 + 1$ real
  variables, representing a distribution $\mu$ over $P^{-1}(1)$, its
  correlation matrix $\rho$, and finally its probability under
  $\Lambda$.  The constraints that each $\rho$ is the correlation
  matrix of the corresponding $\mu$ and that $\Lambda(G_i) = 0$ for
  every $G_i$ can be written as a finite number of polynomial equations
  in the $M(|P^{-1}(1)| + k^2 + 1)$ variables.\footnote{The variables 
  representing probabilities need to be non-negative. This can be effected by taking them to be squares of 
  respective variables.}   In other words the set of
  $m$-vanishing $\Lambda$ of support size $\le M$ form an algebraic set,
  so determining whether such $\Lambda$ exists boils down to
  determining whether this algebraic set is non-empty, which is
  decidable \cite{Tar51}.
\end{proof}

\section{Algorithm}
\label{sec:algorithm}

In this section we give an approximation algorithm with approximation ratio strictly larger than 
$|P^{-1}(1)|/2^k$  for predicates $P$
which do not satisfy the condition of Theorem~\ref{thm:main}.  Thus,
there exists an $m$ such that for every distribution $\Lambda$ (over
$P$-supported correlation matrices), there is an odd-degree multigraph
$G$ over $S \subseteq [k]$ with at most $m$ edges such that
$\hat{P}(S) \ne 0$ and $\Lambda(G) \ne 0$.  For the rest of this
section, fix this value of $m$.

\subsection{Rounding Monomials}
\label{sec:round-layer}

First, we give an algorithm which will allow us to ``pick up" a
contribution proportional to $\rho(G)$ for any monomial, where $\rho$
is the correlation matrix of the local distribution (given by the SDP)
on that monomial and $G$ is any graph.  Lemma~\ref{lemma:round-layer}
below formalizes the high-level idea given in
Section~\ref{sec:proof-ideas}.  Recall that the variables of the CSP
are partitioned into $k$ layers, there are $n$ variables in each
layer, and the SDP relaxation is as in Figure~\ref{fig:sdp}.

\begin{lemma}
  \label{lemma:round-layer}
  Let $S \subseteq [k]$ be a set of layers, $G$ be an odd-degree
  multigraph on $S$, and $\tau > 0$.  Then for all sufficiently large
  $B \ge \poly(\log 1/\tau)$ (where the polynomial depends only on
  $G$) there is a polynomial time algorithm which, given an SDP
  solution as in Figure~\ref{fig:sdp} outputs an assignment $\alpha: S
  \times [n] \rightarrow [-1,1]$ to the layers of $S$ such that the
  following holds.

  Let $V: S \rightarrow [n]$ be any choice of variables, one from each layer in 
  $S$.  Then
  \[
  \E\left[\prod_{i \in S} \alpha_{i,V(i)}\right] =  \frac{\rho(G)}{B^{|S|}} \pm \tau,
  \]
  where $\rho$ is the correlation matrix defined by the SDP solution
  on these $k$ variables, i.e., $\rho_{i_1,i_2} =
  \scalprod{\b{v}_{i_1,V(i_1)}, \b{v}_{i_2,V(i_2)}}$.
\end{lemma}

\begin{proof}
  The algorithm works as follows.  For each edge $e \in E(G)$, pick a
  standard Gaussian vector $\b{g}_e$, independently.  For a vertex $u
  \in V(G)$, let $E(u) \subseteq E(G)$ denote the set of edges
  incident on $u$.  For a variable $x_{i,j}$
  such that $i \in S$, set
  \[
  \beta_{i,j} = \prod_{e \in E(i)} \scalprod{\b{g}_e, \b{v}_{i,j}}.
  \]
  Then, set
  \[
  \alpha_{i,j} = \begin{cases}
    \frac{\beta_{i,j}}{B} & \text{if $|\beta_{i,j}| \le B$} \\
    0 & \text{otherwise}.
  \end{cases}
  \]
  Fix $V: S \rightarrow [n]$ as in the statement.
  Now let us analyze $\E[\prod_{i \in S} \alpha_{i, V(i)}]$.  Note
  that without the truncation when $|\beta_{i,V(i)}|$ exceeds $B$, the
  expectation would be exactly equal to (where in the second step below we use the independence of the Gaussians to move the expectation inside the product)
  \begin{align*}
 \E\left[\frac{1}{B^{|S|}} \prod_{i \in S} \beta_{i, V(i)}\right] &=
   \frac{1}{B^{|S|}} \E\left[ \prod_{i \in S} \prod_{e \in E(i)} 
    \scalprod{\b{g}_e, \b{v}_{i, V(i)}} \right] \\
    &=
   \frac{1}{B^{|S|}} \prod_{(a, b) = e \in E(G)} \E\left[ \scalprod{\b{g}_e, \b{v}_{a, V(a)}} \scalprod{\b{g}_e, \b{v}_{b, V(b)}}\right] \\
    &= \frac{1}{B^{|S|}} \prod_{(a, b) \in E(G)} \scalprod{\b{v}_{a, V(a)}, \b{v}_{b, V(b)}} 
 = \frac{\rho(G)}{ B^{|S|}}
  \end{align*}
  Thus we want to bound the expectation of $\left|\prod_{i \in S}
  \frac{\beta_{i, V(i)}}{B} - \prod_{i \in S} \alpha_{i, V(i)}\right|$ by $
\tau$.
  This can be shown to be of order $m\cdot \exp(-B^{2/m}/2)$, because each
  $\beta_{i, V(i)}$ is a product of at most $|E(G)| \le m$ independent gaussians.
  Thus setting $B$ of order $(\log 1/\tau)^{m/2}$ we get the desired
  error bound.
\end{proof}

\subsection{Setup for the Algorithm}

Let $M \leq 2^k \cdot k^{2m}$ be the number of odd-degree graphs on some
$S \subseteq [k]$ with $\hat{P}(S) \ne 0$ and at most $m$ edges, and let
$G_1, \ldots, G_M$ be these graphs.  Further write $S_t = V(G_t)
\subseteq [k]$ for the vertex set of $G_t$.  For a correlation matrix
$\rho \in \R^{k \times k}$, let $q(\rho)$ be the vector $q(\rho) = (\rho(G_1), \rho(G_2),
\ldots, \rho(G_M)) \in \R^M$, and let $Q \subseteq \R^M$ be the convex
hull of $\{q(\rho): \text{$\rho$ is $P$-supported}\}$.

Note that a $\Lambda$ such that $\Lambda(G_i) = 0$ for all $1 \le i
\le M$ is precisely a convex combination $\Lambda$ of $\rho$'s such that
$\E_{\rho \sim \Lambda}[q(\rho)] = \b{0}$.  In other words since $P$
does not satisfy the condition of Theorem~\ref{thm:main} we have
that $Q$ does not contain the origin.  Furthermore $Q$ is compact and
so we can find a separating hyperplane $(\gamma_1, \ldots, \gamma_M)$
such that $\sum_{t=1}^M \gamma_t \rho(G_t) > \delta$ for every
$P$-supported $\rho$ and some universal constant $\delta$ (depending
only on $P$).

Now let $\tau = \frac{\delta}{4M}$ and set $B = \poly (\log 1/\tau)$
large enough to make Lemma~\ref{lemma:round-layer} work for all the
graphs $G_1, \ldots, G_M$.  In our algorithm, we are going to choose
one $t \in [M]$ at random and then the algorithm is going to focus
solely on the terms involving layers $S_t$.  More precisely, as we
shall see in the next section, $t$ should be chosen with probability
proportional to $\left| \frac{\gamma_t}{\hat{P}(S_t)} \right| \cdot
B^{|S_t|}$.  In order for this to make sense, we therefore need that
\begin{equation}
  \label{eq:alg normalization}
  \sum_{t=1}^M \left| \frac{\gamma_t}{\hat{P}(S_t)} \right| \cdot B^{|S_t|} \le 1.
\end{equation}
Fortunately, we can assume without loss of generality that this holds:
since $B^{|S_t|} \le B^k$ depends sub-linearly (in fact even
poly-logarithmically) on $1/\tau$, dividing each $\gamma_t$ by some
factor $f > 1$ causes $\delta$ and $\tau$ to also be divided by $f$,
which in turn changes $B^{|S_t|}$ to $\poly \log (f/\tau) = o(f
B^{|S_t|})$, so that the sum in the left hand side of \eqref{eq:alg
  normalization} decreases by a factor which is super-constant in $f$.
Hence choosing $f$ a sufficiently large constant, we can make
\eqref{eq:alg normalization} hold.

\subsection{The Rounding Algorithm}
\label{sec:full-algorithm}

We are now ready to describe the algorithm.  Without loss of
generality, we may assume that we are given a
$(1-\epsilon)$-satisfiable $\MaxPartCSP(P)$ instance where $\epsilon >
0$ is some sufficiently small constant (depending on $P$) to be
determined later.  If the instance is not $(1-\epsilon)$-satisfiable
then a random assignment already gives an approximation ratio of
$\frac{|P^{-1}(1)|}{2^k(1-\epsilon)}$.

By Markov's inequality, for at least a $1 - \sqrt{\epsilon}$ fraction
of constraints $(T_i, P_i)$ we have $\E_{x \sim \mu_i}[P_i(x)] \ge 1 -
\sqrt{\epsilon}$.  In other words, $\mu_i$ has a $1-\sqrt{\epsilon}$
fraction of its mass on $P_i^{-1}(1)$.

\begin{claim}
  Given a correlation matrix $\rho \in \R^{k \times k}$ of a
  distribution which is $(1-\sqrt{\epsilon})$-supported on
  $P^{-1}(1)$, there is a $P$-supported correlation matrix $\rho'$
  such that $|\rho(G) - \rho'(G)| \le \sqrt{\epsilon} 2^m$ for every
  $G$ on $m$ edges.
\end{claim}

Thus, setting $\epsilon < \frac{\delta^2}{4 \cdot 2^m}$ we have that $\sum_i
\gamma_i \rho(G_i) > \delta/2$ for all correlation matrices of
distributions which are $(1-\sqrt{\epsilon})$-supported on
$P^{-1}(1)$.

Now the rounding algorithm is as in Figure~\ref{fig:algorithm}.

\begin{figure}[h]
  \framebox{
    \parbox{0.95\textwidth}{
      \begin{enumerate}
      \item Pick $t \in [M]$ with probability $\left|\frac{\gamma_t}{\hat{P}(S_t)}\right|\cdot B^{|S_t|}$.
      \item Using Lemma \ref{lemma:round-layer}, round the variables in layers in $S_t$ using graph $G_t$. For every other layer, set all the variables in
        that layer to $0$.
      \item If $\sign(\gamma_t \hat{P}(S_t)) = -1$, then select an odd sized subset $A$ of $S_t$ at random, else select an even sized subset $A$ of $S_t$
        at random.  Flip the sign of all variables in layers in $A$.
      \end{enumerate}
    }
  }
  \caption{Rounding algorithm for $\MaxPartCSP(P)$}
  \label{fig:algorithm}
\end{figure}

Now, fix the value of $t$ chosen in step 1, and
let $\alpha_{i, j}$ be the rounded value to the variable $x_{i,j}$ as in Lemma \ref{lemma:round-layer} and $\tilde{\alpha}_{i,j}$
be equal to  $\alpha_{i, j}$ or its negation after the third step above.
The assignment $\tilde{\alpha}$ is in $[-1,1]^{n \times k}$ but as the
objective function is multilinear it can be greedily adjusted to an
integral assignment in $\{-1,1\}^{n \times k}$ without decreasing the
objective value, so it suffices to study $\tilde{\alpha}$.
Let $S \subseteq [k]$ be some set of layers and $V: S \rightarrow [n]$ be any choice of
variables from these layers. Observe that for every $S \ne S_t$,
\[
\E_\alpha \left[ \prod_{i \in S} \tilde{\alpha}_{i, V(i)} \right]  =  0.
\]
To see this, note that  if $S \not\subseteq S_t$, then the variables in layers $S \setminus S_t$ are set to $0$.  On the other hand if $S \subset S_t$, then we flip the signs of a random set of layers of either odd or even size.  As the distribution over which layers get flipped is $(|S_t|-1)$-wise independent, the layers of any $S \subset S_t$ get flipped completely uniformly.

On the other hand, by Lemma \ref{lemma:round-layer} and
the way the signs are flipped in the third step, if $S = S_t$ we have
\[
\E_\alpha [ \prod_{i \in S} \tilde{\alpha}_{i, V(i)} ]  = \sign(\gamma_t \hat{P}(S_t)) \cdot \frac{\rho(G_t)}{B^{|S_t|}} \pm \tau.
\]
Thus, taking the expectation of $\E_{\alpha}[\prod_{i \in S} \tilde{\alpha}_{i, V(i)}]$ over $t \in [M]$ chosen according to Step 1, we have
\begin{align*}
\E_t\left[ \E_\alpha [ \prod_{i \in S} \tilde{\alpha}_{i, V(i)} ]\right]  &= \sum_{t: S_t = S} \left|\frac{\gamma_t}{\hat{P}(S_t)}\right| \cdot  B^{|S_t|}
   \left( \sign(\gamma_t \hat{P}(S_t)) \cdot \frac{\rho(G_t)}{B^{|S_t|}} \pm \tau \right) \\
   & = \sum_{t: S_t = S} \frac{\gamma_t}{\hat{P}(S_t)} \rho(G_t) \pm \tau
\end{align*}

Now we can analyze the probability that any specific constraint is
satisfied.  Let $(T_i, P_i)$ be a constraint involving one variable
from each layer which is $(1-\sqrt{\epsilon})$-satisfied by the SDP
solution.
In other words, $\E_{x \sim \mu_i}[P_i(x)] \ge
1-\sqrt{\epsilon}$.  Write $V: [k] \rightarrow [n]$ for the variables
involved (i.e., $T_i = \{(i',V(i')): i' \in [k]\}$) and write $P_i(x)
= P(b_1 x_{1,V(1)}, \ldots, b_k x_{k,V(k)})$ for some signs $b_1,
\ldots, b_k \in \{-1,1\}$.

We also associate the domain of $P_i$ with $\{-1,1\}^k$ in the obvious
way.  As such, it is easy to verify that the Fourier coefficient
$\hat{P_i}(S)$ for $S \subseteq [k]$ satisfies $\hat{P_i}(S) =
\hat{P}(S) \chi_S(b)$.  Furthermore, let $\tilde{\mu_i}$ be the
distribution over $\{-1,1\}^k$ obtained by sampling from $\mu_i$ and
performing coordinatewise multiplication by $b$, and let $\rho$
(resp.~$\tilde{\rho}$) denote the correlation matrix of $\mu_i$
(resp.~$\tilde{\mu_i}$).  Then, for any graph $G$ we have
\[
\rho(G) = \prod_{(a,a') \in E} \rho_{a,a'} = \prod_{(a,a') \in E} \tilde{\rho}_{a,a'} \chi_{\{a,a'\}}(b) = \tilde{\rho}(G) \cdot \chi_{\Odd(G)}(b)
\]
where $\Odd(G)$ denotes the set of odd-degree vertices of $G$.  In
particular for $G = G_t$ all vertices have odd degree so $\Odd(G_t) =
S_t$.  Using this and noting that $\tilde{\mu_i}$ is
$(1-\sqrt{\epsilon})$-supported on satisfying assignments of $P$ we
see that
\[
\sum_t \gamma_t \rho(G_t) \chi_{S_t}(b) = \sum_{t} \gamma_t
\tilde{\rho}(G_t) > \delta /2.
\]
We then have the following, where the expectation below is taken over
all the random choices of the algorithm (including that of $t \in [M]$).
\begin{align*}
\E [ P_i(\tilde{\alpha}_{1,V(1)}, \cdots, \tilde{\alpha}_{k, V(k)}) ]  &=
\hat{P_i}(\emptyset) + \sum_{\substack{\emptyset \ne S \subseteq [k]\\\hat{P_i}(S)\not= 0}}
   \hat{P_i}(S) \E [\prod_{i\in S} \tilde{\alpha}_{i, V(i)}   ] \\
&  =   \hat{P}(\emptyset) + \sum_{\substack{\emptyset \ne S \subseteq [k]\\\hat{P}(S)\not= 0}}
   \hat{P}(S) \chi_S(b)  \sum_{t: S_t = S} \frac{\gamma_t}{\hat{P}(S_t)} \rho(G_t) \pm \tau \\
 &   =   \hat{P}(\emptyset) + \sum_{t=1}^M \gamma_t \rho(G_t) \chi_{S_t}(b) \pm \tau M 
   >    \hat{P}(\emptyset) + \frac{\delta}{4}.
\end{align*}

Thus the total fraction of constraints satisfied by the algorithm is in expectation
at least $(1-\sqrt{\epsilon})(\hat{P}(\emptyset) + \delta/4)$ which is
at least $\hat{P}(\emptyset) + \delta/8$ assuming $\epsilon <
(\delta/8)^2$.

In other words, the algorithm finds a $(\hat{P}(\emptyset) +
\delta/8)$-approximate solution on all instances with value at least
$1-\epsilon$.  Combining this with a random assignment gives an
approximation better than $\hat{P}(\emptyset)$ for any instance, and
concludes the proof of approximability of $\MaxPartCSP(P)$.

\section{Hardness}
\label{sec:hardness}

In this section we show that any $P$ which satisfies the condition of
Theorem~\ref{thm:main} is approximation resistant, assuming the UGC.
As usual, we prove hardness by designing an appropriate dictatorship
test, which is given in Section~\ref{sec:dictatorship test}, followed
by the (standard) hardness reduction in Section~\ref{sec:reduction}.

\newcommand{\M}{\mathcal{M}}
\newcommand{\Tester}{\mathcal{T}}

\subsection{Dictatorship Test}
\label{sec:dictatorship test}

\begin{theorem}\label{thm:dictatortest}
  Let $P$ satisfy the condition of Theorem~\ref{thm:main}.
  Then for every $k$ and $\epsilon > 0$ there exists a $\delta >
  0$ such that the following holds for all $n$.

  There is a randomized algorithm $\Tester$ which, given oracle access
  to $k$ \emph{odd} functions $f_1, \ldots, f_k: \{-1,1\}^n \rightarrow [-1,1]$,
  produces $k$ queries $X_1, \ldots, X_k \in \{-1,1\}^n$ such that
  \begin{description}
  \item[(Yes)] If $f_1(x) = f_2(x) = \ldots = f_k(x) = x_i$ are the
    same dictator function, then $\E[P(f_1(X_1), \ldots, f_k(X_k))] \ge
    1 - \epsilon$.
  \item[(No)] If all $f_i$'s have $\Inf_j^{\le 1/\delta}(f_i) \le \delta$
    for all $j \in [n]$ then $\E[P(f_1(X_1), \ldots, f_k(X_k))] \le
    \hat{P}(\emptyset) + \epsilon$
  \end{description}
\end{theorem}

Let $m = k \cdot n$ and let $\Lambda$ be a distribution over
$P$-supported correlation matrices which is $m$-vanishing on $P$.  Let
$U_k$ be the uniform distribution over $\{-1,1\}^k$.  The tester
$\Tester$ is described in Figure~\ref{fig:tester}.

\begin{figure}[h]
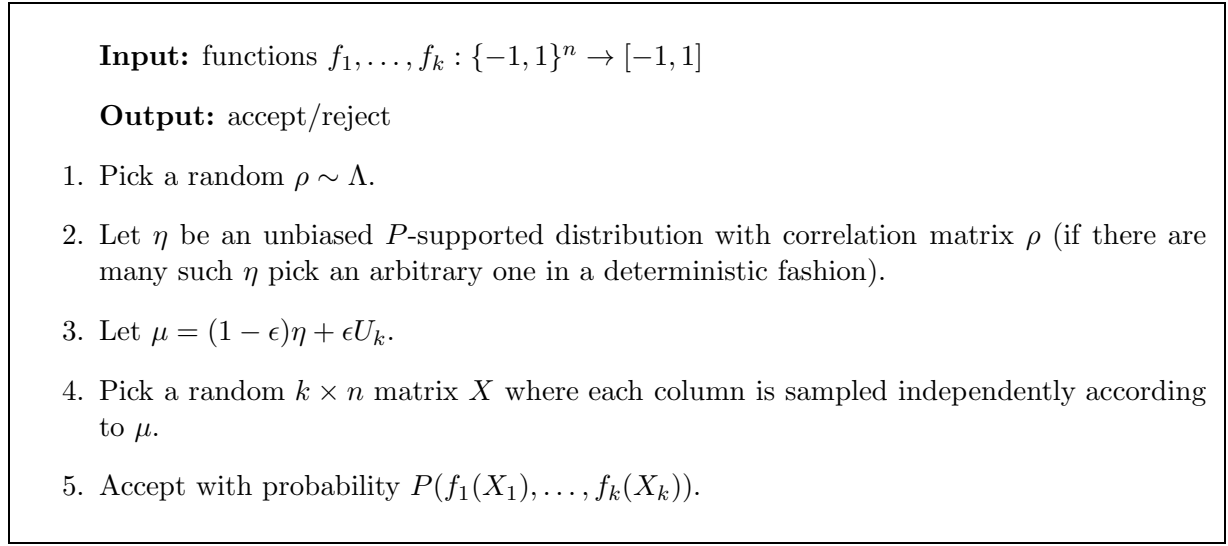

  \framebox{
    \parbox{0.95\textwidth}{
      \begin{enumerate}
      \item[] \textbf{Input:} functions $f_1, \ldots, f_k: \{-1,1\}^n \rightarrow [-1,1]$
      \item[] \textbf{Output:} accept/reject

      \item Pick a random $\rho \sim \Lambda$.
      \item Let $\eta$ be an unbiased $P$-supported
        distribution with correlation matrix $\rho$ (if there are many such $\eta$ pick an arbitrary one in a deterministic fashion).
      \item Let $\mu = (1-\epsilon)\eta + \epsilon U_k$.
      \item Pick a random $k \times n$ matrix $X$ where each column is sampled
        independently according to $\mu$.
      \item Accept with probability $P(f_1(X_1), \ldots, f_k(X_k))$.
      \end{enumerate}
    }
  }
  \caption{Dictatorship Test}
  \label{fig:tester}
\end{figure}

That the completeness is $1-\epsilon$ follows immediately from
$\supp(\eta) \subseteq P^{-1}(1)$.

Let us then analyze the soundness.  The acceptance probability can be written
as
\begin{align}
\Pr[\text{$\Tester$ accepts}] &= \E_{\rho \sim \Lambda}\left[\E_{X \sim \mu}[P(f_1(X_1), \ldots, f_k(X_k))]\right] \nonumber\\
&= \hat{P}(\emptyset) + \E_{\rho}\left[\sum_{S \ne \emptyset} \hat{P}(S) \E_X[\prod_{i\in S} f_i(X_i)]\right].
\label{eqn:test accept prob}
\end{align}
Let $\rho'$ be the correlation matrix of $\mu$.  Note that $\rho' =
(1-\epsilon)\rho + \epsilon I$.
 Fix the value of $\rho$ for the moment, and let
$G$ be a random $k \times n$ matrix of standard Gaussians with the
same covariances as $X$ (i.e., the columns are independent and in the
$j$'th column we have $\E[G_{i,j}G_{i',j}] = \rho'_{i,i'} =
(1-\epsilon)\rho_{i,i'}$ for $i \ne i'$).

Next, set $\delta$ small enough so that Theorem~\ref{thm:invariance}
gives that if $\Inf_j^{\le 1/\delta}(f_i) \le \delta$ for all $i \in
[k]$, $j \in [n]$ then
\begin{equation}
  \label{eqn:approx by gaussians}
\left|\E_X[\prod_{i\in S} f_i(X_i)] - \E_G[\prod_{i\in S} f^{\le 1/\delta}_i(G_i)]\right| \le \epsilon/2^k
\end{equation}
for all $S$.  Define $f'_i = f^{\le 1/\delta}$.
We need to understand expressions of the form $\E_G[\prod_{i \in
    S} f'_i(G_i)]$ for $S \subseteq [k]$.
Expanding $f'_i = \sum_{\substack{T \subseteq[n]}} \hat{f'}(T) \chi_T$ and applying
Lemma~\ref{lemma:gaussproduct}, we obtain
\begin{align}
\E_G\left[\prod_{i \in S} f'_i(G_i)\right] &= \sum_{\{T_i\}_{i \in S}} \prod_{i \in S} \hat{f'_i}(T_i) \prod_{j=1}^n \E\left[\prod_{i: j \in T_i} G_{ij}\right] \nonumber\\
&= \sum_{\{T_i\}_{i \in S}} \prod_{i \in S} \hat{f'_i}(T_i) \prod_{j=1}^n \sum_{M_j \in \M(\{i: j \in T_i\})} \prod_{i,i' \in M_j} \rho'_{i,i'},
\label{eqn:prodexp}
\end{align}
where we write $\M(S)$ for the set of perfect matchings on the complete
graph with vertex set $S$.

\newcommand{\bT}{\b{T}}
\newcommand{\bM}{\b{M}}

For a choice $\bT = \{T_i\}_{i \in S}$ of $T_i$'s, let $c(\bT) =
\prod_{i \in S} \hat{f'_i}(T_i)$.  Further, for a choice of matchings
$\bM = (M_1, \ldots, M_n)$ let $H(\bM)$ denote the multigraph being
the union of $M_1, \ldots, M_n$.  With a slight abuse of notation,
write $\M(\bT)$ for the set of $\bM$'s for a given $\bT$; i.e.,
$\M(\bT) = \{(M_1, \ldots, M_n): M_j \in \M(\{i: j \in T_i\})$.  With
all this cumbersome notation in place, the equation above simplifies
to
\begin{align}
  \eqref{eqn:prodexp} &= \sum_{\bT} \sum_{\bM \in {\bf\M}(\bT)} c(\bT) \rho'(H(\bM)),
  \label{eqn:prodexp2}
\end{align}
where $\bT$ ranges over all $\{T_i \subseteq [n]\}_{i \in S}$.  Note
that since $f_i$ (and therefore also $f'_i$) is odd, we can restrict the sum to $\bT$ such that
each $|T_i|$ is odd, implying that $H(\bM)$ is always odd degree.

Plugging \eqref{eqn:prodexp2} into \eqref{eqn:approx by gaussians}, we have
\[
\left|\E_X[\prod_{i\in S} f_i(X_i)] - \sum_{\bT} \sum_{\bM \in {\bf\M}(\bT)} c(\bT) \rho'(H(\bM))\right| \le \epsilon/2^k
\]
Finally, plugging this into \eqref{eqn:test accept prob} and using the
identity $\rho' = (1-\epsilon)\rho + \epsilon I$ yields
\begin{align*}
\Pr[\text{$\Tester$ accepts}] &\le \hat{P}(\emptyset) + \E_{\rho}\left[\sum_{S \ne \emptyset} \left( \hat{P}(S) \sum_{\bT} \sum_{\bM \in {\bf\M}(\bT)} c(\bT) \rho'(H(\bM)) + \epsilon/2^k\right) \right] \\
&\le \hat{P}(\emptyset) + \epsilon + \sum_{S \ne \emptyset} \hat{P}(S) \sum_{\bT} \sum_{\bM \in {\bf\M}(\bT)} c(\bT) (1-\epsilon)^{|E(H(\bM))|} \E_{\rho}[\rho(H(\bM))] \\
& = \hat{P}(\emptyset) + \epsilon,
\end{align*}
where the last equality follows by the $m$-vanishing property of
$\Lambda$: $H(\bM)$ has at most $n |S|/2 < nk = m$ edges, and so for
each $S$ either $\hat{P}(S) = 0$ or $\E_{\rho}[\rho(H(\bM))] =
\Lambda(H(\bM)) = 0$.

\subsection{Hardness Reduction}
\label{sec:reduction}

Given the dictatorship test as in Theorem \ref{thm:dictatortest}, a UGC-based hardness reduction
can be designed in a standard manner. Some care needs to be taken however to ensure that the
CSP instance produced by the reduction is $k$-partite. As is standard, we present the
reduction as a Probabilistically Checkable Proof (PCP) for NP whose acceptance predicate matches the
predicate $P$, has completeness $1-o(1)$ and soundness $\frac{|P^{-1}(1)|}{2^k} + o(1)$.

The PCP is based on the conjectured NP-hard instance  $\Lambda = (U, V, E, \Pi, [L])$ of Unique Games
as in Definition \ref{def:ug}. Let  $L$ and $\gamma$ be as in
Conjecture \ref{conj:ugc}.  The PCP proof consists of $k$ layers where the bits in the
$i$'th layer correspond to $V_i \times \{-1,1\}^L$ and $V_i$ is a copy of the ``right hand side"
$V$ of the UG instance. For any $v_i \in V_i (= V)$,
the set of bits $\{ v_i \} \times \{-1,1\}^L$ correspond to the bits of the long code
of the label of $v_i$. In a ``correct" proof, the assignment to these bits corresponds to a
dictatorship function $f(x) = x_j$ where $j \in [L]$ is the intended label of $v_i$.

For a function $g:\{-1,1\}^L \rightarrow \{-1,1\}$ and a permutation $\pi: [L] \rightarrow [L]$, let
$g \circ \pi^{-1}: \{-1,1\}^L \rightarrow \{-1,1\}$ denote the function defined as
$g \circ \pi^{-1} (x) = g (x_{\pi^{-1}(1)}, \ldots, x_{\pi^{-1}(L)})$.
The PCP verifier
proceeds as in Figure~\ref{fig:verifier}.

\begin{figure}[h]
  \framebox{
    \parbox{0.95\textwidth}{
      \begin{enumerate}
      \item Pick a random vertex $u \in U$.
      \item Pick $k$ random neighbors of $u$, namely $v_1, \ldots, v_k \in V$.
      \item Let $g_1, \ldots, g_k$ be the functions (supposed long codes) for $v_1 \in V_1, \ldots,
        v_k \in V_k$ respectively.
      \item Let $f_1, \ldots, f_k$ be the permuted versions of $g_1, \ldots, g_k$ respectively, i.e.,
        $f_i = g_i \circ \pi_i^{-1}$, $\pi_i = \pi_{e_i = (u, v_i)}$ for $1 \leq i \leq k$.
      \item Run the dictatorship test as in Theorem \ref{thm:dictatortest} on $(f_1,\ldots,f_k)$.
      \end{enumerate}
    }
  }
  \caption{PCP Verifier}
  \label{fig:verifier}
\end{figure}

\subsubsection{Completeness}

Let $\ell: U \cup V \rightarrow [L]$ be a labeling to the UG instance that satisfies $1-\gamma$ fraction of its edges. For every $v_i \in V_i (= V)$, let $g_i$ be the long code of $\ell(v_i)$, i.e. $g_i(x) = x_{\ell(v_i)}$. With probability at least $1- k \gamma$, all $k$ edges $(u,v_1), \ldots, (u, v_k)$ are
satisfied by the labeling and whenever this holds, the dictatorship test accepts with probability $1-\epsilon$. The latter conclusion follows by observing that if $\pi_i (\ell(u) ) = \ell(v_i)$
for every $1 \leq i \leq k$, then in the PCP test above,
$$f_i (x) =  g_i \circ \pi_i^{-1} (x)  =
g_i( x_{\pi_i^{-1}(1)}, \ldots, x_{\pi_i^{-1}(L)} ) =  x_{\pi_i^{-1} ( \ell(v_i))} = x_{\ell(u)},$$
and hence $f_1, \ldots, f_k$ are identical dictatorship functions.

\subsubsection{Soundness}

Assume that the soundness of the UG instance is at most $\gamma$ which is chosen to be sufficiently small
beforehand. Fix any layer $i$ in the PCP proof and let $g_{i,v}$ be the supposed long code corresponding
to the vertex $v$ (in the $i$'th layer). For any $u \in U$, define the function $f_{i,u}$ which is the
average of functions over the neighbors of $u$ after appropriate permutation:
$$ f_{i,u} (x) =
   \E_{v: (u,v) \in E} \left[  g_{i,v}  \circ \pi_{(u,v)}^{-1}  (x) \right].$$
Note that $f_{i,u}$ are $[-1,1]$-valued. By a standard argument, we may assume that for
all but $\sqrt{\gamma}$ fraction of $u \in U$, the function $f_{i,u}$ has no coordinate that has
degree $1/\delta$ influence $\delta$ (referred to as a low-influence function for brevity).

Otherwise, suppose that for $\sqrt{\gamma}$ fraction of
$u$, $f_{i,u}$ has a  coordinate that has degree  $1/\delta$ influence $\delta$. For brevity, call any such coordinate simply as an influential
coordinate. The set of all influential coordinates has size bounded by
$1/\delta^2$. Assign this bounded set as the set of candidate labels for $u$.  For any
influential coordinate
$j \in [L]$, since $f_{i,u}$ is an average of $g_{i, v} \circ \pi_{(u,v)}^{-1}$
over neighbors of $u$, by an averaging argument, for at least $\delta/2$ fraction of the neighbors,
$\pi_{(u,v)}(j)$ is influential for $g_{i,v}$. All influential coordinates of $g_{i,v}$ are
assigned as the candidate labels for $v$. Now define a (randomized) labeling that selects one label at random from the candidate set of each vertex. The argument sketched implies  that this labeling satisfies
$\sqrt{\gamma} \cdot \delta/2 \cdot \delta^4$ fraction of the UG edges. This is a contradiction if the
soundness $\gamma$ was chosen to be sufficiently small to begin with.

Hence except with probability $k\sqrt{\gamma}$, the PCP verifier chooses $u \in U$ such that
the $k$ functions $f_{i,u}$, one in each layer, are all low influence functions. Whenever this holds, the analysis of the
dictatorship test implies that the verifier accepts with probability at most $\frac{|P^{-1}(1)|}{2^k} + \epsilon$. One only needs to observe that for a fixed $u$, the verifier picks its random neighbor
in each layer and hence running the test on these random neighbors (one in each layer)
has the same effect as running the test
on the (possibly non-boolean) averaged functions (again, one in each layer). Formally, fixing $u$,
\begin{multline*}
\E_{v_1, \ldots, v_k} \left[ \E_\Tester [  P(g_{1,v_1}\circ \pi_1^{-1} (X_1), \ldots,
    g_{k,v_k}\circ \pi_k^{-1} (X_k)    )   ]     \right]
  \\
  = \E_\Tester \left[  P( \E_{v_1} [g_{1,v_1}\circ \pi_1^{-1} (X_1)], \ldots, \E_{v_k} [g_{k,v_k}\circ \pi_k^{-1} (X_k)]            )               \right] \\
  = \E_\Tester \left[  P (f_{1,u}(X_1), \ldots, f_{k, u}(X_k)   )    \right].
 \end{multline*}

\section{Acknowledgements} 

We are grateful to Johan H\aa stad for many insightful discussions throughout this work, and to the anonymous referees for their helpful suggestions.



\bibliographystyle{alpha}
\bibliography{resist}

\begin{appendix}

\end{appendix}
\end{document}